\documentclass[11pt]{article}

\usepackage{amsfonts,amsmath,amssymb}
\usepackage{amsthm}
\usepackage{bm}
\usepackage{xspace}
\usepackage{times}
\usepackage{verbatim}
\usepackage[margin=1in]{geometry}
\usepackage{algpseudocode}
\usepackage[dvipsnames,usenames]{xcolor}

\newtheorem{theorem}{Theorem}
\newtheorem{lemma}[theorem]{Lemma}
\newtheorem{corollary}[theorem]{Corollary}
\newtheorem{remark}{Remark}[theorem]
\newtheorem*{conjecture}{Conjecture}

\newtheorem{definition}{Definition}

\newcommand{\calD}{\ensuremath{\mathcal D}\xspace}
\newcommand{\calQ}{\ensuremath{\mathcal Q}\xspace}
\newcommand{\calX}{\ensuremath{\mathcal X}\xspace}
\newcommand{\calY}{\ensuremath{\mathcal Y}\xspace}

\newcommand{\R}{\ensuremath{\mathbb{R}}\xspace}
\newcommand{\N}{\ensuremath{\mathbb{N}}\xspace}

\newcommand{\eps}{\ensuremath{\epsilon}\xspace}

\newcommand{\mecpure}{\ensuremath{M_\eps}\xspace}
\newcommand{\mecapprox}{\ensuremath{M_{\eps, \delta}}\xspace}

\newcommand{\meclonepure}{\ensuremath{M^{\ell_1}_\eps}\xspace}
\newcommand{\mecloneapprox}{\ensuremath{M^{\ell_1}_{\eps, \delta}}\xspace}

\newcommand{\query}{\calQ}
\newcommand{\database}{\calD}
\newcommand{\metric}{\calX}
\newcommand{\metricy}{\calY}
\newcommand{\dist}{d}
\newcommand{\hamming}[2]{\|#1 - #2\|_H}

\newcommand{\funci}{F^{(i)}_\database}

\newcommand{\bound}[1]{{{#1}}}

\newif\iffullversion
\fullversiontrue

\begin{document}

\title{Exploiting Metric Structure for Efficient Private Query Release}


\author{Zhiyi Huang\thanks{Computer and Information Science, University of Pennsylvania. Email: {\tt hzhiyi@cis.upenn.edu}.}
\and
Aaron Roth\thanks{Computer and Information Science, University of Pennsylvania. Email: {\tt aaroth@cis.upenn.edu}.}}


\maketitle

\begin{abstract}
We consider the problem of privately answering queries defined on databases which are collections of points belonging to some metric space. We give simple, computationally efficient algorithms for answering \emph{distance queries} defined over an arbitrary metric. Distance queries are specified by points in the metric space, and ask for the average distance from the query point to the points contained in the database, according to the specified metric. Our algorithms run efficiently in the database size and the dimension of the space, and operate in both the online query release setting, and the offline setting in which they must in polynomial time generate a fixed data structure which can answer \emph{all} queries of interest. This represents one of the first subclasses of linear queries for which \emph{efficient} algorithms are known for the private query release problem, circumventing known hardness results for generic linear queries.
\end{abstract}

\section{Introduction}
Consider an online retailer who is attempting to recommend products to customers as they arrive. The retailer may have a great deal of demographic information about each customer, both from cookies and from data obtained from tracking networks. Moreover, the retailer will also have information about what other, demographically similar customers have purchased in the past. If the retailer can identify which cluster of customers the new arrival most resembles, then it can likely provide a useful set of recommendations. Note that this problem reduces to computing the average \emph{distance} from the new arrival to past customers in each demographic cluster, where the distance metric may be complex and domain specific.\footnote{Note that the most natural metric for this problem may not be defined by an $\ell_p$ norm, but may be something more combinatorial, like edit distance on various categorical features.}

For legal reasons (i.e. to adhere to it's stated privacy policy), or for public relations reasons, the retailer may not want the recommendations given to some customer $i$ to reveal information about any specific past customer $j \neq i$. Therefore, it would be helpful if the retailer could compute these distance queries while guaranteeing that these computations satisfy \emph{differential privacy}. Informally, this means that the distances computed from each new customer to the demographic clusters should be insensitive in the data of any single user in the database of past customers.

Distance queries are a subclass of \emph{linear queries}, which are well studied in the differential privacy literature \cite{BLR08,DNRRV09,DRV10,RR10,HR10}. For example, the data analyst could answer $k$ such queries from an $\ell$-dimensional metric space, on a database of size $n$ using the private multiplicative weights mechanism of Hardt and Rothblum \cite{HR10} with error that scales as $O(\textrm{poly}(\log(k),\ell)/\sqrt{n})$.\footnote{All of the mechanisms for answering linear queries \cite{BLR08,DNRRV09,RR10,DRV10,RR10,HR10,GHRU11,GRU12} are defined over \emph{discrete} domains $X$ and have an error dependence on $\log |X|$. In contrast, these queries are defined over \emph{continuous} $\ell$-dimensional domains, and so it is not clear that this previous work even applies. However, metric queries are Lipschitz, and so these mechanisms can be run on a discrete grid with roughly $n^{\Omega(\ell)}$ points, giving a polynomial dependence on $\ell$ in the error bounds, but an exponential dependence on $\ell$ in the running time.} However, none of these mechanisms is computationally efficient, and even for the best of these mechanisms, the running time \emph{per query} will be exponential in $\ell$, the dimension of the space. What's more, there is strong evidence that there do not exist computationally efficient mechanisms that can usefully and privately answer more than $O(n^2)$ general linear queries \cite{DNRRV09,Ull12}. A major open question in differential privacy is to determine whether there exist interesting subclasses of linear queries for which efficient algorithms do exist.

In this paper, we show that distance queries using an arbitrary metric are one such class. We give simple, efficient algorithms for answering exponentially many distance queries defined over any metric space with bounded diameter. In the online query release setting, our algorithms run in time nearly linear in the dimension of the space and the size of the private database per query. Our algorithms remain efficient even in the \emph{offline} query release setting, in which the mechanism must in one shot (and with only polynomial running time) privately generate a synopsis which can answer all of the (possibly exponentially many) queries of interest. This represents one of the first high dimensional classes of linear queries which are known to have \emph{computationally efficient} private query release mechanisms which can answer large numbers of queries.

\subsection{Our Techniques}
At a high level, our mechanism is based on the reduction from \emph{online learning algorithms} to \emph{private query release mechanisms} developed in a series of papers \cite{RR10,HR10,GHRU11,GRU12}. Specifically, we use the fact that an online mistake-bounded learning algorithm for learning the function $F:C\rightarrow \mathbb{R}$, which maps queries $f \in C$ to their answers $f(D)$ on the private database $D$ generically gives the existence of a private query release mechanism in the interactive setting, where the running time per query is equal to the update time of the learning algorithm.

We observe that when the queries are \emph{metric distance} queries over some continuous $\ell_p$ metric space \metric, then $F:\metric\rightarrow\mathbb{R}$ is a convex, Lipschitz-continuous function. Motivated by this observation, we give a simple mistake-bounded learning algorithm for learning arbitrary convex Lipschitz-continuous functions over the unit interval $[0,1]$ by approximating $F$ by successively finer piecewise linear approximations. Our algorithm has a natural generalization to the $\ell$-dimensional rectangle $[0,1]^\ell$, but unfortunately the mistake bound of this generalization necessarily grows exponentially with $\ell$.

Instead, we observe that if $\metric = [0,1]^\ell$, and is endowed with the $\ell_1$ metric, then $F$ can be decomposed into $\ell$ $1$-dimensional functions $F_1,\ldots,F_\ell$ each defined only over the unit interval $[0,1]$. Hence, for the $\ell_1$ metric, our learning algorithm can be extended to $[0,1]^\ell$ with only a linear increase in the mistake bound. In other words, \emph{the $\ell_1$ metric is an easy metric for differential privacy}. In fact, for $\ell_1$ distance queries, our algorithm achieves per-query error $O(\textrm{poly}(\log(k),\ell)/n^{4/5})$, improving on the worst-case error guarantees that would be given by inefficient generic query release mechanisms like \cite{BLR08,HR10}.

Finally, we show that our algorithm can be used to answer distance queries for any metric space that can be embedded into poly$(\ell)$-dimensional $\ell_1$ space using a \emph{low sensitivity embedding}. A sensitivity $s$ embedding is one that maps any pair of databases that differ in only $1$ element into a pair of projected databases that differ in only $s$ entries. Oblivious embeddings, such as the almost-isometric embedding from $\ell_2$ into $\ell_1$ are 1-sensitive \cite{FLM77, I06}. On the other hand, generic embeddings, such as the embedding from an arbitrary metric space into $\ell_1$ that follows from Bourgain's theorem can have sensitivity as high as $n$ \cite{B85, LLR95}.

We observe, however, that for our purposes, we do not require that the embedding preserve distances between pairs of database points, or between pairs of query points, but rather only between database points and query points. Therefore, we are able to prove a variant of Bourgain's theorem, which only preserves distances between query points and database points. This gives a $1$-sensitive embedding from \emph{any} metric space into $\log k$ dimensional $\ell_1$ space, with distortion $\log k$, which works for any collection of $k$ distance queries. In particular, this gives us an efficient offline algorithm for answering $k$ distance queries defined over an \emph{arbitrary} bounded diameter metric that has multiplicative error $O(\log k)$ and additive error $O(\textrm{polylog}(k)/n^{4/5})$. Our use of metric embeddings is novel in the context of differential privacy, and we believe that they will be useful tools for developing efficient algorithms in the future as we identify other privacy-friendly metrics in addition to $\ell_1$.
\subsection{Related Work}
Differential privacy was developed in a series of papers \cite{DN03,BDMN05,DMNS06}, culminating in the definition by Dwork, Mcsherry, Nissim, and Smith \cite{DMNS06}. It is accompanied by a vast literature which we do not attempt to survey.

Dwork et al.~\cite{DMNS06} also introduced the \emph{Laplace} mechanism, which together with the composition theorems of Dwork, Rothblum, and Vadhan \cite{DRV10} gives an efficient, interactive method for privately answering nearly $n^2$ arbitrary low-sensitivity queries on a database of size $n$ to non-trivial accuracy. On the other hand, it has been known since Blum, Ligett, and Roth \cite{BLR08} that it is information theoretically possible to privately answer nearly exponentially many \emph{linear} queries to non-trivial accuracy, but the mechanism of \cite{BLR08} is not computationally efficient. A series of papers \cite{BLR08,DNRRV09,DRV10,RR10,HR10,GHRU11,GRU12} has extended the work of \cite{BLR08}, improving its accuracy, running time, and generality. The state of the art is the private multiplicative weights mechanism of Hardt and Rothblum \cite{HR10}. However, even this mechanism has running time that is linear in the size of the \emph{data universe}, or in other words \emph{exponential} in the dimension of the data. Finding algorithms which can achieve error bounds similar to \cite{BLR08,HR10} while running in time only polynomial in the size of the database and the data dimension has been a major open question in the differential privacy literature since at least \cite{BLR08}, who explicitly ask this question.

Unfortunately, a striking recent result of Ullman \cite{Ull12}, building on the beautiful work of Dwork, Naor, Reingold, Rothblum, and Vadhan \cite{DNRRV09}, shows that assuming the existence of one way functions, no polynomial time algorithm can answer more than $O(n^2)$ arbitrary linear queries. In other words, the Laplace mechanism of \cite{DMNS06} is nearly optimal among all computationally efficient algorithms for privately answering queries at a comparable level of generality. This result suggests that to make progress on the problem of computationally efficient private query release, we must abandon the goal of designing mechanisms which can answer \emph{arbitrary} linear queries, and instead focus on classes of queries that have some particular structure that we can exploit.

Before this work, there were very few efficient algorithms for privately releasing classes of ``high dimensional'' linear queries with worst case error guarantees. Blum, Ligett, and Roth \cite{BLR08} gave efficient algorithms for two low dimensional classes of queries: constant dimensional axis aligned rectangles, and large margin halfspaces\footnote{Note that halfspace queries are in general high dimensional, but the large-margin assumption implies that the data has intrinsic dimension only roughly $O(\log n)$, since the dimensionality of the data can be reduced using the Johnson-Lindenstrauss lemma without affecting the value of any of the halfspace predicates.}. Feldman et al. gave efficient algorithms for releasing Euclidean $k$-medians queries in a constant dimensional unit ball \cite{FFKN09}. Note that when we restrict our attention to Euclidean metric spaces, our queries correspond to $1$-median queries. In contrast to \cite{FFKN09}, we can handle arbitrary metrics, and our algorithms are efficient also in the dimension of the metric space. Blum and Roth \cite{BR11} gave an efficient algorithm for releasing linear queries defined over predicates with extremely sparse truth tables, but such queries are very rare. Only slightly more is known for \emph{average case} error. Gupta et al.~\cite{GHRU11} gave a polynomial time algorithm for releasing the answers (to linear, but non-trivial error) to conjunctions, where the error is measured in the average case on conjunctions drawn from a product distribution. Hardt, Rothblum, and Servedio \cite{HRS12} gave a polynomial time algorithm for releasing answers to parity queries, where the error is measured in the average case on parities drawn from a product distribution. Although it is known how to convert average case error to worst-case error using the private boosting technique of Dwork, Rothblum, and Vadhan \cite{DRV10}, the boosting algorithm itself is not computationally efficient when the class of queries is large, and so cannot be applied in this setting where we are interested in polynomial time algorithms. For the special case of privately releasing conjunctions in $\ell$ dimensions, Thaler, Ullman, and Vadhan \cite{TUV12}, building on the work of Hardt, Rothblum, and Servedio \cite{HRS12}, give an algorithm that runs in time $O(2^{\sqrt{\ell}})$, improving on the generic bound of $O(2^\ell)$. Finding a polynomial time algorithm for releasing conjunctions remains an open problem.

Metric embeddings have proven to be a useful technique in theoretical computer science, particularly when designing approximation algorithms. See \cite{PiotrSurvey} for a useful survey. The specific embeddings that we use in this paper are the nearly isometric embedding from $\ell_2$ into $\ell_1$ using random projections  \cite{FLM77, I06}, and a variant of Bourgain's theorem \cite{B85, LLR95}, which allows the embedding of an \emph{arbitrary} metric into $\ell_1$. Our use of metric embeddings is slightly different than its typical use in approximation algorithms. Typically, metric embeddings are used to embed some problem into a metric in which some optimization problem of interest is tractable. In our case, we are embedding metrics into $\ell_1$, for which the \emph{information theoretic} problem of query release is simpler, since a $d$ dimensional $\ell_1$ metric can be decomposed into $d$ $1$-dimensional metric spaces. On the one hand, for privacy, we have a stronger constraint on the type of metric embeddings we can employ: we require them to be \emph{low sensitivity embeddings}, which map neighboring databases to databases of bounded distance (in the hamming metric). The embedding corresponding to Bourgain's theorem does not satisfy this property. On the other hand, we do not require that the embedding preserve the distances between pairs of database points, or pairs of query points, but merely between query points and database points. This allows us to prove a variant of Bourgain's theorem that is $1$-sensitive. We think that metric embeddings may prove to be a useful tool in the design of efficient private query release algorithms, and in particular, identifying other privacy friendly metrics, and the study of other low sensitivity embeddings is a very interesting future direction.  

\section{Preliminaries}

\subsection{Model}

Let $(\metric, \dist)$ be an arbitrary metric space. Let $\database \in \metric^{n}$ be a database consists of $n$ points in the metric space. For the sake of presentation, we will focus on metric spaces with diameter $1$ through out the main body of this paper. This is simply a matter of scaling: all of our error bounds hold for arbitrary diameter spaces, with a linear dependence on the diameter. 

We will consider the problem of releasing distance queries while preserving the privacy of the elements in the database, where each query is a point $y \in \metric$ in the metric space and the answer for a given query $y$ is the average distance from $y$ to the elements in the database, i.e., $\sum_{x \in \database} \frac{1}{n}\dist(x, y)$. Let $\query \in \metric^{k}$ be the set of distance queries asked by the data analyst. We will let $\database(\query) \in \R^k$ denote the exact answer to the queries $\query$ with respect to database $\database$. We will usually use $x_i$'s to denote data points and $y_j$'s to denote query points.

\iffullversion
\paragraph{Query Release Mechanisms}
\else
\medskip
\noindent{\bf Query Release Mechanisms~} 
\fi
We will consider two settings for query release in this paper:
The first setting is the {\em interactive setting}, where the queries are not given upfront but instead arrive online. An {\em interactive query release mechanism} needs to provide an answer for each query as it arrives. The answer can depend on the query, the private database, and the state of the mechanism, but not on future queries. An interactive query release mechanism is said to be efficient if the \emph{per-query} running time is polynomial in the database size $n$ and the dimension of the metric space $\ell$.

The second setting is the {\em non-interactive setting}. A {\em non-interactive query release mechanism} takes the database as input and outputs an algorithm that can answer {\em all} queries without further access to the database. We say an offline query release mechanism is efficient if both the running time of the mechanism and the running time per query of the algorithm it constructs are polynomial in $n$ and $\ell$.

\subsection{Differential Privacy}

We let $\hamming{\database_1}{\database_2}$ denote the hamming distance between two databases $\database_1$ and $\database_2$. Two databases are adjacent if the hamming distance between them is at most $1$ (i.e. they differ in a single element). We will write $n = |\database|$ to denote the size of the database. We will consider the by now standard privacy solution concept of ``differential privacy'' \cite{DMNS06}.

\begin{definition}[$(\eps, \delta)$-Differential Privacy]
  A mechanism $M$ is {\em $(\eps, \delta)$-differentially private} if for all adjacent databases $\database_1$ and $\database_2$, any  set of queries $\query$, and for all subsets of possible answers $S \subset \R^k$, we have
  $$\Pr\left[M(\database_1, \calQ) \in S\right] \le \exp(\eps) \Pr\left[M(\database_2, \calQ) \in S\right] + \delta \enspace.$$
  If $\delta = 0$, then we say that $M$ is $\eps$-differentially private.
\end{definition}

 A function $f:\metric^n\rightarrow \mathbb{R}$ is said to have sensitivity $\Delta$ with respect to the private database if $\max_{\database_1,\database_2}|f(\database_1) - f(\database_2)| \leq \Delta$, where the max is taken over all pairs of adjacent databases.

When we talk about the privacy of interactive mechanisms, the range of the mechanism is considered to be the entire transcript of queries and answers communicated between the data analyst and the mechanism (see \cite{DRV10,HR10} for a more precise formalization of the model). An interactive mechanism is $(\eps, \delta)$-differential private if the probability that the transcript falls into any chosen subset differs by at most an $\exp(\eps)$ multiplicative factor and a $\delta$ additive factor for any two adjacent databases.


Given a mechanism, we will measure its accuracy in terms of answering distance queries as follows.

\begin{definition}[Accuracy]
  A mechanism $M$ is {\em $(\alpha, \beta)$-accurate} if for any database $\database$ and any set of queries $\query$, with probability at least $1-\beta$, the mechanism answers every query up to an additive error $\alpha$, i.e.,
  $$\Pr\left[\|M(\database, \query) - \database(\query)\|_{\infty} \le \alpha \right] \ge 1-\beta \enspace.$$
\end{definition} 

\section{Releasing $\ell_1$-Distance Queries}

In this section, we consider $\ell_1$ distance queries, i.e., we let $\metric \subset [0, 1]^\ell$ and $d = \|.\|_1$ such that the diameter of $\metric$ (with respect to~$\ell_1$) is $1$. We present private, computationally efficient mechanisms for releasing the answers to $\ell_1$ distance queries in both the interactive and offline setting. These mechanisms for releasing $\ell_1$ distances will serve as important building blocks for our results for other metrics. First, let us formally state our result in the interactive setting:

\begin{theorem} \label{thm:l1-interactive}
  There is an interactive $(\eps, \delta)$-differentially private mechanism for releasing answers to distance queries with respect to~$(\metric, \|.\|_1)$ that is $(\alpha, \beta)$-accurate with $\alpha$ satisfying
  $$\alpha = O \left( \frac{\ell^{9/5} \log^{4/5} (4 / \delta) \log^{4/5} (k / \beta)}{n^{4/5} \epsilon^{4/5}} \right) \enspace.$$
  There is also an interactive $\epsilon$-differentially private mechanism for releasing distance queries with respect to~$(\metric, \|.\|_1)$ that is $(\alpha, \beta)$-accurate for $\alpha$ satisfying
  $$\alpha = O \left( \frac{\ell^{7/3} \log^{2/3} (k / \beta)}{n^{2/3} \eps^{2/3}} \right) \enspace.$$
  The per-query running times of both mechanisms is $O(\ell n)$ per query.
\end{theorem}

As an extension of the above theorem, we also get the following result in the offline setting.

\begin{theorem} \label{thm:l1-offline}
  There is a poly-time $(\eps, \delta)$-differentially private offline mechanism that is $(\alpha, \beta)$-accurate for releasing distance queries with respect to~$(\metric, \|.\|_1)$, for $\alpha$ satisfying
  $$\alpha = O \left( \frac{\ell^{9/5} \log^{4/5} (4 / \delta) \log^{4/5} (n \ell / \beta)}{n^{4/5} \epsilon^{4/5}} \right) \enspace.$$
  There is a poly-time $\epsilon$-differentially private offline mechanism that is $(\alpha, \beta)$-accurate for releasing distance queries with respect to~$(\metric, \|.\|_1)$, for $\alpha$ satisfying
  $$\alpha = O \left( \frac{\ell^{7/3} \log^{2/3} (n \ell / \beta)}{n^{2/3} \eps^{2/3}} \right) \enspace.$$
  The total running time of this mechanism is $O\left(\frac{\ell^3 n^2}{\alpha}\right)$.
\end{theorem}
\begin{remark}Note that the offline mechanism has no dependence on the number of queries asked, in either the accuracy or the running time. It in one shot produces a data structure that can be used to accurately answer \emph{all} $\ell_1$ queries.
\end{remark}



\iffullversion
\paragraph{Proof Overview}
\else
\medskip
\noindent{\bf Proof Overview~} 
\fi
To prove Theorem \ref{thm:l1-interactive}, we will use the connection between private query release and online learning, which was established in \cite{RR10, HR10, GRU12, JT12}. We will briefly review this connection in Section \ref{sec:learning}. Based on this connection, it suffices to provide an online learning algorithm that learns the function mapping queries to their answers with respect to~the database using a small number of updates. Next, we will shift our viewpoint by interpreting each database as a $1$-Lipschitz and convex function that maps the query points to real values between $[0, 1]$. The structure of the $\ell_1$ metric allows us to reduce the problem to learning $\ell$ different one dimensional $1$-Lipschitz and convex functions, for which we propose in Section \ref{sec:convex} an online learning algorithm that only requires $O(\alpha^{-1/2})$ updates to achieve an additive error bound $\alpha$. Finally, we combine these ingredients to give an interactive differentially private mechanism for releasing answers for $\ell_1$ distance queries in Section \ref{sec:proofl1} and complete the proof of Theorem \ref{thm:l1-interactive}. Roughly speaking, the interactive mechanism will always maintain a hypothesis function that maps queries to answers and it will update the hypothesis function using the online learning algorithm whenever the hypothesis function makes a mistake. Finally, we show that there is an explicit set of $O(\ell^2 / \alpha)$ queries such that asking these queries to the interactive mechanism is sufficient to guarantee that the hypothesis function is accurate with respect to~all queries. So Theorem \ref{thm:l1-offline} follows because the offline mechanism can first ask these queries to the interactive mechanism and then release the hypothesis function.


\iffullversion
\subsection{Query Release from Iterative Database Construction} \label{sec:learning}
\else
\subsection{Query Release from Iterative Database\\ Construction} \label{sec:learning}
\fi

In this section, we give a (variant) of the definition of the iterative construction framework defined in \cite{GRU12}, generalizing the median mechanism and the multiplicative weights mechanism \cite{RR10,HR10}. Let $F_C:C\rightarrow \mathbb{R}$ be the function such that for each $y \in C$, $F_C(y) = \database(y)$: i.e. $F$ maps queries to their answers evaluated on $\database$. Note that $F_C(y)$ is a $1/n$ sensitive function in the private database. The variant of the definition of Iterative Database Construction that we give allows the learning algorithm to also learn the answers to some set $S$ of $O(1/n)$ sensitive functions on $F_C$ as well (in addition to just $F_C(y)$). In our application, $S$ will consist of queries about the \emph{derivative} of $F_C$, where in our case, $F_C$ will be a (one-sided) differentiable function.

\begin{definition}[\cite{RR10,HR10,GRU12}]
  Let $F_C:C\rightarrow \mathbb{R}$ be the function such that for each $y \in C$, $F_C(y) = \database(y)$: i.e. $F$ maps queries to their answers evaluated on $\database$. Let $S = \{f_1,\ldots,f_{|S|}\}$ be a collection of functions $f_i:\metric\times \database\rightarrow \mathbb{R}$ that are each $O(1/n)$ sensitive in the private database $\database$. Given an error bound $\alpha > 0$ and an error tolerance $c$, an {\em iterative database construction algorithm} using functions $S$ with respect to a class of queries $C$ plays the following game with an adversary: 
  \begin{enumerate}
    \setlength{\partopsep}{0pt}
    \setlength{\topsep}{0pt}
    \setlength{\parsep}{0pt}
    \setlength{\itemsep}{0pt}
    \item The algorithm maintains a hypothesis function $\hat{F_t}:C\rightarrow \mathbb{R}$ on which it can evaluate queries, which is initialized to be some default function $\hat{F_0}$ at step $0$.
    \item In each step $t \ge 1$, the adversary (adaptively) chooses a query $y_t \in \metric$, at which point the algorithm predicts a query value $\hat{F_t}(y_t)$. If $|\hat{F_t}(y_t) - F_C(y_t)| > \alpha$, then we say the algorithm has {\em made a mistake}. At this point, the algorithm receives $|S|$ values $a_1,\ldots,a_{|S|} \in \mathbb{R}$ such that for all $i$: $a_i \in [f_i(y_t,\database) - c\alpha, f_i(y_t,\database) + c\alpha]$. The algorithm may update its hypothesis function using this information.
  \end{enumerate}
\end{definition}


\begin{definition}[Mistake Bound]
  An iterative database construction algorithm has a {\em mistake bound} $m : \R_+ \mapsto \N_+$, if for any given error bound $\alpha$, no adversary can (adaptively) choose a sequence of queries to force the algorithm to make $m(\alpha)+1$ mistakes.
\end{definition}

\begin{lemma}[\cite{RR10,HR10,GRU12}] \label{lem:idc}
  If there is an iterative database construction using functions $S$ for releasing a query class $C$ with mistake bound $m(\alpha)$ with respect to some error tolerance $c$, then there is an $(\epsilon, \delta)$-differentially private mechanism in the interactive setting that is $(\alpha, \beta)$-accurate for answering queries $C$, for $\alpha$ satisfying
  $$c\alpha = \frac{1}{n \eps} 3000 \sqrt{|S|m(\alpha)} \log(4 / \delta) \log(k / \beta) \enspace.$$
  There is also an $\epsilon$-differentially private mechanism in the interactive setting that is $(\alpha, \beta)$-accurate, for $\alpha$ satisfying
  $$c\alpha = \frac{1}{n\eps} 3000 ~ |S|\cdot m(\alpha) \log(k / \beta) \enspace.$$
  Moreover, the per-query running time of the query release mechanism is equal to (up to constant factors) the running time of the per-round running time of the iterative database construction algorithm.
\end{lemma}


\iffullversion
\paragraph{Representing $\ell_1$ Databases as Decomposable Convex Functions}
\else
\medskip
\noindent{\bf Representing $\ell_1$ Databases as Decomposable Convex Functions~} 
\fi
Consider a database $\database$ where the universe is the $\ell$-dimensional unit cube $\metric = [0,1]^\ell$ endowed with the $\ell_1$ metric. In this setting, the function mapping queries $y \in \metric$ to their answers takes the form: $F_{\database}(y) = \frac{1}{n}\sum_{x \in \database}||x-y||_1$, which is a $1/n$-Lipschitz convex function of $y$. We wish to proceed by providing an iterative database construction for $\ell_1$ distance queries using these properties. Observe that because we are working with the $\ell_1$ metric, we can write: $F_{\database}(y) = \sum_{i=1}^\ell F^{(i)}_{\database}(y)$, where $F^{(i)}_{\database}(y) = \frac{1}{n}\sum_{x \in \database}|x_i-y_i|$. Observe that each function $F^{(i)}_{\database}(y)$ is $1$-Lipschitz and convex, and has a 1-dimensional range $[0,1]$. Therefore, to learn an approximation to $F_{\database}(y)$ up to some error $\alpha$, it suffices to learn an approximation to each $F^{(i)}_{\database}(y)$ to error $\alpha/\ell$. This is the approach we take.

\subsection{Learning 1-Lipschitz Convex Functions} \label{sec:convex}

In this section we study the problem of iteratively constructing an arbitrary continuous, $1$-Lipschitz, and convex function $G : [0, 1] \mapsto [0, 1]$ up to some additive error $\alpha_1$ with noisy oracle access to the function. Here, the oracle can return the function value $G(x)$ and the derivative $G'(x)$ given any $x \in [0, 1]$ up to an additive error of $\alpha_1/4$. Here, we assume the derivative $G'$ is well defined in $[0, 1]$: If $G$ is not differentiable at $x$, then we assume the derivative $G'(x)$ is (consistently) defined to be any value between the left and right derivatives at $x$.


We will first present an algorithm that learns any one-dimension $1$-Lipschitz and convex function using an exact oracle. Then, we will explain why this algorithm is in fact noise-tolerant. Finally, we show that this result naturally extends to multi-dimensional decomposable functions.

\iffullversion
\paragraph{Learning 1-D Functions with an Accurate Oracle}
\else
\medskip
\noindent{\bf Learning 1-D Functions with an Accurate Oracle~} 
\fi
We will consider maintaining a hypothesis piece-wise linear function $\hat{G}(x)$ via the algorithm given in Figure \ref{fig:1}. We will analyze the number of updates needed by this algorithm before it has learned a piece-wise linear function $\hat{G}$ that approximates $G$ everywhere up to additive error $\alpha_1$.

\begin{figure}
  \centering
  \fbox{
  \iffullversion
  \begin{minipage}{.98\textwidth}
  \else
  \begin{minipage}{.46\textwidth}
  \fi
    {\bf Learning a $1$-Lipschitz and Convex Function}

    \medskip

    Maintain $\hat{G}(x) = \max_k \{ a_k \cdot x + b_k\}$ where $a_k \in [-1, 1]$ and $b_k \in \R$ define a set of linear functions.

    \medskip

    \noindent{\bf Initial Step:~} Let $a_0 = 0$ and $b_0 = 0$.

    \medskip

    \noindent{\bf Update Step $t \ge 1$:~} While the update generator returns a distinguishing point $x^*_t$, we shall add the tangent line at $x^*_t$ with respect to function $g$ to the set of linear functions, i.e., $a_t = G'(x^*_t)$ and $b_t = G(x^*_t) - G(x^*_t)' \cdot x^*_t$.
  \end{minipage}}
  \caption{An algorithm for learning a $1$-Lipschitz and convex one-dimensional function by approximating it with a piece-wise linear function. The algorithm always predicts according to $\hat{G}$. When it makes a mistake, it is given an update point $x^*_t$ together with $G(x^*_t)$ and $G(x^*_t)'$.} \label{fig:1}
\end{figure}

First, for any $1$-Lipschitz (possibly non-convex) function $G$ and any given error bound $\alpha_1 > 0$, the algorithm in Figure \ref{fig:1} will make at most $1/\alpha_1$ mistakes. This is because the function being $1$-Lipschitz implies that the tangent line at each update point $x^*_t$ is a good approximation (up to error $\alpha_1$) in the neighborhood $[x^*_t - \alpha_1, x^*_t + \alpha_1]$, and hence any pair of update points are at least $\alpha_1$ away from each other. Further, it is easy to construct examples where this bound is tight up to a constant.

Next, we will show that using the convexity of function $G$, we can improve the mistake bound to $O(\frac{1}{\sqrt{\alpha_1}})$.

\begin{lemma} \label{lem:learning}
  For any $1$-Lipschitz convex function $G$ and any given error bound $\alpha_1 \in (0, 1)$, the algorithm in Figure \ref{fig:1} will make at most $\frac{3}{\sqrt{\alpha_1}}$ updates.
\end{lemma}

\begin{proof}
  Consider any two update points $x^*_t$ and $x^*_{t'}$. Let us assume w.l.o.g.~that $t < t'$. Then, by our assumption, the tangent line at $x^*_t$ does not approximate the function value of $f$ at $x^*_{t'}$ up to an additive error of $\alpha$. Therefore, we get that
  \begin{eqnarray}
    \alpha_1 & < & G(x^*_{t'}) - \left( G'(x^*_t) (x^*_{t'} - x^*_t) + G(x^*_t) \right) \notag \\
    & = & G(x^*_{t'}) - G(x^*_t) - G'(x^*_t) (x^*_{t'} - x^*_t) \notag \\
    & \le & G'(x^*_{t'})(x^*_{t'} - x^*_t) - G'(x^*_t) (x^*_{t'} - x^*_t) \notag \\
    & = & (G'(x^*_{t'}) - G'(x^*_t)) (x^*_{t'} - x^*_t) \enspace, \label{eq:1}
  \end{eqnarray}
  where the second inequality is by the convexity of $f$.

  Next, consider a maximal set of update points in sorted order: $-1 \le \hat{x}_1 < \dots < \hat{x}_T \le 1$. Since $G$ is convex and $1$-Lipschitz, we have that $-1 \le G'(\hat{x}_1) < \dots < G'(\hat{x}_T) \le 1$. Therefore, we get that
  \begin{eqnarray*}
    2 \cdot 1 & \ge & (G'(\hat{x}_T) - G'(\hat{x}_1))(\hat{x}_T - \hat{x}_1) \\
    & = & \textstyle \sum_{t=1}^{T-1} (G'(\hat{x}_{t+1}) - G'(\hat{x}_t)) \sum_{i=1}^{T-1} (\hat{x}_{t+1} - \hat{x}_t) \\
    & \ge & \textstyle \left( \sum_{t=1}^{T-1} \sqrt{(G'(\hat{x}_{t+1}) - G'(\hat{x}_t))(\hat{x}_{t+1} - \hat{x}_t)} \right)^2 \\
    & \ge & \textstyle \left( \sum_{t=1}^{T-1} \sqrt{\alpha_1} \right)^2 \enspace.
  \end{eqnarray*}
  Here, the first inequality is by $G'(x_t) \in [-1, 1]$ and $x_t \in [0, 1]$ for $t = 1, \dots, T$; the second inequality is a simple application of the Cauchy-Schwartz inequality; the last inequality is by equation \eqref{eq:1}. So by the above inequality, the number of mistakes is at most $T \le \frac{\sqrt{2}}{\sqrt{\alpha_1}} + 1 < \frac{3}{\sqrt{\alpha_1}}$.
\end{proof}

\iffullversion
\paragraph{Learning 1-D Functions with a Noisy Oracle}
\else
\medskip
\noindent{\bf Learning 1-D Functions with a Noisy Oracle~}
\fi
Note that the domain of the function is $[0, 1]$. So if the tangent line at $x'$ approximates the function value at $x$ up to additive error $\frac{\alpha_1}{2}$, i.e.,
$$G(x) - G(x') + G'(x') (x - x') \le \frac{\alpha_1}{2} \enspace,$$
then a noisy version of the tangent line $\bar{G}(x') + \bar{G}'(x') (x - x')$, where $\bar{G}(x') \in [G(x') - \frac{\alpha_1}{4}, G(x') + \frac{\alpha_1}{4}]$ and $\bar{G}'(x') \in [G'(x') - \frac{\alpha_1}{4}, G'(x') + \frac{\alpha_1}{4}]$, will approximate the value at $x$ up to additive error $\alpha_1$. Hence, the mistake bound of the algorithm in Figure \ref{fig:1} for learning a $1$-Lipschitz and convex function up to additive error $\alpha_1$ using a noisy oracle is no more than the mistake bound for learning the same function up to additive error $\frac{\alpha_1}{2}$ with an accurate oracle. Hence, the mistake bound is still of order $O(\frac{1}{\sqrt{\alpha_1}})$.

\begin{lemma} \label{lem:learning2}
  For any $1$-Lipschitz convex function $g$ and any given error bound $\alpha_1 \in (0, 1)$, the algorithm in Figure \ref{fig:1} will make at most $O(\frac{1}{\sqrt{\alpha_1}})$ updates with an $\frac{\alpha_1}{4}$-noisy oracle.
\end{lemma}

\iffullversion
\paragraph{Learning Decomposable Functions}
\else
\medskip
\noindent{\bf Learning Decomposable Functions~}
\fi
Suppose we want to learn an $\ell$-dimension decomposable convex function $F_\database = \sum_{i=1}^\ell \funci$ up to additive error $\alpha$, where each $\funci$ is convex and $1$-Lipschitz. Then, it suffices to learn the $1$-Lipschitz convex functions $F^{(i)}_D$ for each coordinate up to error $\alpha_1 = \frac{\alpha}{\ell}$. So as a simple corollary of Lemma \ref{lem:learning2}, we have the following lemma:

\begin{lemma}
  \label{lem:learning3}
For any function $F_{\database}:[0,1]^\ell\rightarrow \mathbb{R}$ such that:
\begin{enumerate}
\item $F_{\database}(y) = \sum_{i=1}^{\ell}F_{\database}^{(i)}(y_i)$ where each $F_{\database}^{(i)}:[0,1]\rightarrow [0,1]$ is $1$-Lipschitz and convex, and:
\item For every $y \in [0,1]^\ell$ and every $i \in [\ell]$: $F_{\database}^{(i)}(y)$ and $(F_{\database}^{(i)}(y))'$ are $1/n$-sensitive in $\database$
\end{enumerate}
there is an iterative database construction algorithm for $F_{\database}$ using a collection of $2\ell$ functions $S$ with respect to an error tolerance $1/(4\ell)$ that has a mistake bound of $m(\alpha) = O(\ell^{3/2}/\alpha^{1/2})$.
\end{lemma}

\begin{proof}
  Let $\alpha_1 = \frac{\alpha}{\ell}$. Consider the following algorithm:
  \begin{enumerate}
    \setlength{\partopsep}{0pt}
    \setlength{\topsep}{0pt}
    \setlength{\parsep}{0pt}
    \setlength{\itemsep}{0pt}
    \item The algorithm maintains a hypothesis function $\hat{F}_t = \sum_{i=1}^{\ell}\hat{F}_t^{(i)}$ by maintaining $\ell$ one-dimension piecewise-linear hypothesis functions $\hat{F}^{(i)}_t : [0, 1] \mapsto [0, 1]$ for each $i \in [\ell]$ via the one-dimension learning algorithm with error tolerance $\alpha_1$, and letting $\hat{F}_t = \sum_{i=1}^\ell \hat{F}^{(i)}_t$.
    \item If the algorithm makes a mistake on query $y_t \in [0, 1]^\ell$, then the algorithm asks query $y_t$ to each of the one-dimensional learning algorithms. On any of the one-dimensional learning algorithms $i$ on which a mistake is made, the algorithm queries two values: $F^{(i)}_t(y_{ti})$ and $(F^{(i)}_t)'(y_{ti})$, tolerating additive error up to $\alpha_1 / (4)$, and updates the hypothesis $\hat{F}^{(i)}_t$, $i = 1, \dots, \ell$, accordingly using the one dimensional learning algorithm. Note that this leads to at most $|S| = 2\ell$ queries per update.
  \end{enumerate}

  Note that whenever the above algorithm makes a mistake, at least one of the one-dimensional algorithms must also make a mistake (since otherwise the total error was at most $\ell\alpha_1 = \alpha$), and therefore we can charge this mistake to the mistake bound of at least one of the one-dimensional learning algorithms. By Lemma \ref{lem:learning2}, the number of times that the hypothesis function $\hat{F}^{(i)}_t$ in each coordinate admits additive error at least $\alpha_1$ is at most $O(1/\sqrt{\alpha_1})$. So the above iterative database construction algorithm has mistake bound $O(\ell/\sqrt{\alpha_1}) = O(\ell^{3/2}/\alpha^{1/2})$.
\end{proof}


\subsection{Proofs of Theorem \ref{thm:l1-interactive} and Theorem \ref{thm:l1-offline}} \label{sec:proofl1}

\iffullversion
\begin{proof}[Proof of Theorem \ref{thm:l1-interactive}]
\else
\begin{proof}[of Theorem \ref{thm:l1-interactive}]
\fi
  Since the $\ell_1$ distance function in each coordinate has range $[0, 1]$ and is $1$-Lipschitz, we get that $F^{(i)}_t$ and the derivative $(F^{(i)}_t)'$ are $O(1/n)$-sensitive. So by Lemma \ref{lem:learning3}, there is an iterative database construction algorithm for releasing answers to $\ell_1$ distance queries that uses a set $S$ of $2 \ell$ $O(1/n)$-sensitive queries with error $\alpha_1$, and the algorithm has mistake bound $O(\ell^{3/2}/\alpha^{1/2})$.

  By plugging the parameters of the above iterative database construction algorithm to Lemma \ref{lem:idc}, we get that there is an $(\epsilon, \delta)$-differentially private mechanism in the interactive setting that is $(\alpha, \beta)$-accurate for releasing distance queries with respect to~metric space $([0, 1]^\ell, \|.\|_1)$, for $\alpha$ satisfying
  $$\frac{\alpha}{\ell} = O\left(\frac{1}{n \eps} \sqrt{\frac{\ell^{5/2}}{\alpha^{1/2}}} \log(4 / \delta) \log(k / \beta)\right) \enspace.$$
  Solving the above we get that
  $$\alpha = O \left( \frac{\ell^{9/5} \log^{4/5} (4 / \delta) \log^{4/5} (k / \beta)}{n^{4/5} \epsilon^{4/5}} \right) \enspace.$$

  We also get that there is an $\epsilon$-differentially private mechanism in the interactive setting that is $(\alpha, \beta)$-accurate, for $\alpha$ satisfying
  $$\frac{\alpha}{\ell} = O \left( \frac{1}{n\eps} \frac{\ell^{5/2}}{\alpha^{1/2}} \log(k / \beta) \right)\enspace.$$
  Solving the above we get that
  $$\alpha = O \left( \frac{\ell^{7/3} \log^{2/3} (4 / \delta) \log^{2/3} (k / \beta)}{n^{2/3} \epsilon^{2/3}} \right) \enspace.$$

  The analysis of the running time per query is straightforward and hence omitted.
\end{proof}

\iffullversion
\begin{proof}[Proof of Theorem \ref{thm:l1-offline}]
\else
\begin{proof}[of Theorem \ref{thm:l1-offline}]
\fi
Consider running the online query release mechanism with accuracy $\alpha' = \alpha/2$. To give an offline mechanism, we simply describe a fixed set of $\ell/\alpha'$ queries that we can make to each of the $\ell$ one-dimensional learning algorithms maintaining $\hat{F}^{(i)}_{\database}$ that guarantees that for each $y \in [0,1]$, $|\hat{F}^{(i)}_{\database}(y) - F^{(i)}_{\database}(y)| \leq \alpha/\ell$ Once we have this condition, we know that for each $y \in [0,1]^\ell$: $|\hat{F}_{\database}(y) - F_{\database}(y)| \leq \alpha$. The queries are simple: we just take our query set to be a grid: $T = \{0, \alpha'/\ell, 2\alpha'/\ell, 3\alpha'/\ell, \ldots, 1\}$.By the guarantees of the $1$-dimensional learning algorithm, we have that for every $y \in T$, $|\hat{F}^{(i)}_{\database}(y) - F^{(i)}_{\database}(y)| \leq \alpha'/\ell$. Moreover, by the fact that $\hat{F}^{(i)}_{\database}$ is $1$-Lipschitz, and for every $y \in [0,1]$, $d(y,T) \leq \alpha'/\ell$, we have that for every $y \in [0,1]$, $|\hat{F}^{(i)}_{\database}(y) - F^{(i)}_{\database}(y)| \leq 2\alpha'/\ell = \alpha/\ell$, which is the condition we wanted. In total, we make $2\ell^2/\alpha$ queries, and the theorem follows by instantiating the guarantees of the online mechanism with $k = 2\ell^2/\alpha$.
\end{proof}

\iffullversion
\section{Releasing Arbitrary Distance Queries via 1-Sensitive Metric Embeddings}
\else
\section{Releasing Arbitrary Distance\\ Queries via 1-Sensitive Metric\\ Embeddings}
\fi
In this section, we will discuss how to release answers to distance queries with respect to other metric spaces. Our approach is to reduce the problem to releasing answers to $\ell_1$ distance queries via metric embeddings. Recall that an {\em embedding} from a metric space $(\metric, \dist)$ to another metric space $(\metricy, \dist')$ is a mapping $\pi : \metric \mapsto \metricy$. The usefulness of an embedding is measured by how much the embedding distorts the distance between any pair of points.


Note that for the purpose of answering distance queries, the usual definition of distortion is too strong in the sense that the usual notion of distortion considers the worst case distortion for every pair of points in the metric space while we only need to preserve the distances between every data-query pair. So in this paper, we will consider the following weaker notion of expansion, contraction, and distortion of metric embeddings.

\begin{definition}
  Recall that $(\metric, \dist)$ is the metric space of the distance query release problem and $\database$ and $\query$ are the set of data points and the set of query points respectively. The {\em expansion} of an embedding $\pi$ from $(\metric, \dist)$ to another metric space $(\metricy, \dist')$ is
  $$\max_{x \in \metric, y \in \query} \frac{\dist'(\pi(x), \pi(y))}{\dist(x, y)} \enspace.$$
  The {\em contraction} of the embedding is
  $$\max_{x \in \metric, y \in \query} \frac{\dist(x, y)}{\dist'(\pi(x), \pi(y))} \enspace.$$
  The {\em distortion} of the embedding is the product of its expansion and contraction.
\end{definition}

In the rest of this section, we will choose the target metric $(\metricy, \dist')$ to be the $\ell_1$ metric space $([0, 1]^\ell, \|.\|_1)$ and we will always scale the embedding such that the expansion is $1$.

\subsection{1-Sensitive Metric Embeddings}

Suppose we are given such an embedding from $(\metric, \dist)$ to $([0, 1]^{\ell}, \|.\|_1)$ with expansion $1$ and contraction $C$. In some cases, the dimension $\ell$ of the target $\ell_1$ space may depend on the contraction $C$. We will embed both the data points and the query points into $([0, 1]^{\ell}, \|.\|_1)$ and release distance queries via the the mechanism for $\ell_1$. Concretely, consider the mechanisms $\mecapprox$ and $\mecpure$ given in Figure \ref{fig:generic}.

\begin{figure}
  \centering
  \fbox{
  \iffullversion
  \begin{minipage}{.98\textwidth}
  \else
  \begin{minipage}{.46\textwidth}
  \fi
    {\bf Releasing distance queries via embedding into $\ell_1$}
    \medskip

    {\bf Input:~} A set of data points $\database$. A set of queries points $\query$. A 1-sensitive embedding $\pi$ from $(\metric, \dist)$ to $([0, 1]^\ell, \|\cdot\|_1)$.
    \begin{enumerate}
      \setlength{\partopsep}{0pt}
      \setlength{\topsep}{0pt}
      \setlength{\parsep}{0pt}
      \setlength{\itemsep}{0pt}
      \item Construct a {\em proxy database} $\database'$ for releasing $\ell_1$ distances by letting $\pi(x) \in \database'$ for every $x \in \database$.
      \item Use the $(\eps, \delta)$-differentially private mechanism (resp., $\eps$-differentially private mechanism) for releasing $\ell_1$ distances queries to answer $\frac{1}{n} \sum_{x \in \database} \|\pi(x) - \pi(y)\|_1$ for every $y \in \query$ and release them as the answers to $\frac{1}{n} \sum_{x \in \database} \dist(x, y)$ respectively.
    \end{enumerate}
  \end{minipage}}
  \caption{An $(\eps, \delta)$-differentially private mechanism $\mecapprox$ (resp., $\eps$-differentially private mechanism $\mecpure$) for releasing distance queries via embedding into $\ell_1$} \label{fig:generic}
\end{figure}

Let us first consider the accuracy of these mechanisms. The mechanisms will lose a multiplicative factor due to the embedding and an additive factor due to answering the $\ell_1$ queries privately. More precisely,

\begin{theorem} \label{thm:accuracy}
  If the embedding $\pi$ has expansion $1$ and contraction $C$, and if we use the $(\eps, \delta)$-differentially private mechanism to release answers for $\ell_1$ distance queries, then with probability at least $1 - \beta$ the mechanism $\mecapprox$ answers every distance query $y \in \query$ with accuracy
  $$\frac{1}{C} ~ \sum_{x \in \database} \dist(x, y) - \alpha_{\eps, \delta} \le \mecapprox (y, D) \le \sum_{x \in \database} \dist(x, y) + \alpha_{\eps, \delta} \enspace,$$
  where $\bound{\alpha_{\eps, \delta} = \tilde{O} \left( \frac{\ell^{9/5}}{n^{4/5} \eps^{4/5}} \right)}$. If we use the $\eps$-differentially private mechanism to release answers for $\ell_1$ distance queries, then with probability at least $1-\beta$ the mechanism $A_\eps$ answers every query $y \in \query$ with accuracy
  $$\frac{1}{C} \sum_{x \in \database} \dist(x, y) - \alpha_\eps \le A_\eps (\database, y) \le \sum_{x \in \database} \dist(x, y) + \alpha_\eps \enspace,$$
  where \bound{$\alpha_\eps = O \left( \frac{\ell^{7/3}}{n^{2/3} \eps^{2/3}} \right)$}.
\end{theorem}

\begin{remark}
  \label{remark:accuracy}
  If the embedding is nearly isometric, i.e., we can achieve contraction $1+\alpha$ for any small $\alpha > 0$ by embedding in to an $\ell(\alpha)$-dimension $\ell_1$ space, then we will choose the optimal additive error bound such that
  \bound{$$\alpha_{\eps, \delta} = \tilde{O} \left( \frac{\ell(\alpha_{\eps, \delta})^{9/5}}{n^{4/5} \eps^{4/5}} \right) \enspace.$$}
  and
  \bound{$$\alpha_\eps = \tilde{O} \left( \frac{\ell(\alpha_\eps)^{7/3}}{n^{2/3} \eps^{2/3}} \right) \enspace.$$}
\end{remark}

\begin{proof}
  Let us prove the error bound for $\eps$-differential privacy. The proof of the error bound for $(\eps, \delta)$-differential privacy is similar.  We will view the embedding $\pi$ as from $(X, d)$ to $(\pi(X), \|.\|_1)$. Since the embedding $\pi$ has expansion $1$, the image $\pi(X)$ of $X$ has diameter $1$ as well. Let $\meclonepure$ denote the $\eps$-differentially private mechanism for releasing answers to the $\ell_1$ distance queries. Then we have that
  \begin{eqnarray*}
    \mecpure(y, D) & = & \meclonepure(\pi(y), D') \\
    & \le & \sum_{x \in D} \|\pi(x) - \pi(y)\|_1 + \bound{\tilde{O} \left( \frac{\ell^{7/3}}{n^{2/3} \eps^{2/3}} \right)} \\
    & \le & \sum_{x \in D} d(x, y) + \bound{\tilde{O} \left( \frac{\ell^{7/3}}{n^{2/3} \eps^{2/3}} \right)} \enspace.
  \end{eqnarray*}
  The proof of the lower bound is similar, hence omitted.
\end{proof}

Next we will turn to the privacy guarantee of the mechanism. Since we are using either an $(\eps, \delta)$-differentially private mechanism or an $\eps$-differentially private mechanism for releasing answers to the $\ell_1$ distance queries with respect to~the proxy database $D'$, it suffices to ensure that the embeddings of neighboring databases remain neighboring databases. In general, the embedding of some point $x$ may be defined in terms of other data points $y$, which would violate this condition. Formally, we want our embeddings to be $1$-sensitive:

\begin{definition}
  An embedding $\pi$ from $(\metric, \dist)$ to $([0, 1]^\ell, \|.\|_1)$ is {\em $1$-sensitive} if changing a data point $x_i \in \database$ will only change the embedding $\pi(x_i)$ of $x_i$ and will not affect the embedding $\pi(x_j)$ of other $x_j \in \database$ for any $j \ne i$.
\end{definition}

\begin{theorem} \label{thm:privacy}
  If the embedding $\pi$ is $1$-sensitive, and if we use the $(\eps, \delta)$-differentially private mechanism (resp., $\eps$-differentially private mechanism) for releasing answers to the $\ell_1$ distance queries, then the mechanism $\mecapprox$ (resp., $\mecpure$) is $(\eps, \delta)$-differentially private (resp., $\eps$-differentially private).
\end{theorem}

\begin{proof}
  For any two neighboring databases $\database_1$ and $\database_2$, the resulting proxy databases $\database'_1$ and $\database'_2$ in Figure \ref{fig:generic} will either be the same or be neighboring databases since the embedding $\pi$ is $1$-sensitive. Since we are using an $\eps$-differentially private mechanism for releasing $\ell_1$ distances over the proxy databases, we get that for any set of queries $\query$ and for any subset $S$ of possible answers,
  \begin{align*}
    \Pr[\mecpure (\database_1, \query) \in S] & = \Pr[\meclonepure (\database'_1, \pi(\query)) \in S] \\
    & \le \exp(\eps) \Pr[\meclonepure (\database'_2, \pi(\query)) \in S] \\
    & = \Pr[\mecpure (\database_2, \query) \in S] \enspace.
  \end{align*}
  So mechanism $\mecpure$ is $\eps$-differentially private. The proof for $(\eps, \delta)$-differential privacy is similar and hence omitted.
\end{proof}

\begin{remark}
  In principle, we can also consider $s$-sensitive embeddings for small $s$. However, we are not aware of any useful embeddings of this kind. So we will focus on $1$-sensitive embeddings in this paper.
\end{remark}

\begin{remark}
  \label{remark:generic}
  If the embedding $\pi$ is independent of the set $\query$ of queries, then the mechanisms in Figure \ref{fig:generic} can be made interactive or non-interactive by using the interactive or non-interactive mechanisms respectively for releasing answers to the $\ell_1$ distance queries. If the embedding is a function of the query set, then the mechanism will be non-interactive, because potentially all of the queries may be needed to construct the embedding of the database.
\end{remark}

\subsection{Releasing Euclidean Distance via an Oblivious Embedding}

Let us consider releasing distance queries with respect to Euclidean distance. From the metric embedding literature we know that there exists an almost isometric embedding from $\ell_2$ to $\ell_1$. More precisely,

\begin{lemma}[E.g., \cite{FLM77, I06}]
  \label{lem:lp}
  There is an embedding $\pi$ from $([0, 1]^\ell, \|.\|_2)$ to $([0, 1]^{\ell'}, \|.\|_1)$ with expansion $1$, contraction $1+\alpha$, and $\ell' = O \left( \frac{\ell \log(1/\alpha)}{\alpha^2} \right)$. Further, this embedding can be probabilistically constructed in polynomial time by defining each coordinate as a random projection.
\end{lemma}

Since the above embedding is based on random projections, it is $1$-sensitive and independent of the set $\query$ of queries. Thus we can plug this embedding into our framework in Figure \ref{fig:generic} and the following theorem for releasing Euclidean distances follows from Theorem \ref{thm:accuracy}, Remark \ref{remark:accuracy}, Theorem \ref{thm:privacy}, and Remark \ref{remark:generic}. 

\begin{theorem} \label{thm:l2}
  Suppose $(\metric, \|.\|_2)$ is a subspace of the $\ell_2$ space with diameter $1$. Then, there are polynomial time interactive and non-interactive mechanisms for releasing answers to the $\ell_2$ distance queries that are $(\eps, \delta)$-differentially private and $(\alpha_{\eps, \delta}, \beta)$-accurate for $\alpha_{\eps, \delta}$ satisfying
  \bound{$$\alpha_{\eps, \delta} = \tilde{O} \left( \frac{\ell^{9/23}}{n^{4/23} \eps^{4/23}} \right) \enspace.$$}
  There are also polynomial time interactive and non-interactive mechanisms for releasing answers to the $\ell_2$ distance queries that are $\epsilon$-differentially private and $(\alpha_\eps, \beta)$-accurate for $\alpha_\eps$ satisfying
  \bound{$$\alpha_\eps = \tilde{O} \left( \frac{\ell^{7/17}}{n^{2/17} \eps^{2/17}} \right) \enspace.$$}
  The omitted poly-log factors depends on $\ell$, $n$, and $\beta^{-1}$ (and $\delta^{-1}$ for $(\eps, \delta)$-differential privacy) in the offline setting. In the interactive setting, this factor also depends on $\log k$. We remark again that in the \emph{offline} setting, the constructed data structure can answer \emph{all} $\ell_2$ queries.
\end{theorem}






  We remark that Lemma \ref{lem:lp} also holds for $\ell_p$ metrics for $p \in (1, 2)$ (E.g., \cite{FLM77}). So the results stated in Theorem \ref{thm:l2} also apply to $\ell_p$ metrices for $p \in (1, 2)$. Details are omitted.

\subsection{Releasing Distances for General Metric via Bourgain's Theorem}

In this section, we will consider releasing distance queries with respect to an arbitrary metric $(\metric, \dist)$ by embedding it into an $\ell_1$ metric. Bourgain's theorem (e.g., \cite{B85, LLR95}) suggests that for any $m$ points in the metric space, there is an embedding into an $O(\log^2 m)$-dimensional $\ell_1$ space with distortion $O(\log m)$. Unfortunately, this embedding is not oblivious and does not have low sensitivity. However, recall that for the purpose of releasing distance queries, we only need to preserve the distances between all data-query pairs. In other words, it is okay to have the distances between data points (and likewise, between query points) to be highly distorted. Further, we show that for this weaker notion of embedding, there is a variation of Bourgain's theorem using an embedding that is oblivious to the data points, and hence has sensitivity $1$.

Concretely, we will consider the embedding given in Figure \ref{fig:bourgain}. The idea is to define the embedding only using the query points and we will show this is enough to preserve the distances from any point in the metric space to the query points with high probability. Formally, we will prove the following theorem.

\begin{figure}
  \centering
  \fbox{
  \iffullversion
  \begin{minipage}{.98\textwidth}
  \else
  \begin{minipage}{.46\textwidth}
  \fi
    {\bf 1-sensitive variant of Bourgain: Embedding an arbitrary metic space $(X, d)$ into $\ell_1$}

    \medskip

    \noindent{\bf Pre-processing:~} For $1 \le i \le \log k$ and $1 \le j \le K$, where $K$ is chosen to be $512 (\log k + \log n)$, choose a random subset $S_{ij}$ of the query points by picking each query point $y$ independently with probability $2^{-(i-1)}$.

    \medskip

    \noindent{\bf Embedding:~} Given $x$ in the metric space $(X, d)$, embed it into $\{\pi_{ij}(x)\}_{0 \le i \le \log k, 1 \le j \le K}$ in the $O(K \log k)$-dimension $\ell_1$ space by letting $\pi_{ij}(x) = \frac{1}{K \log k} d(x, S_{ij})$.
  \end{minipage}}
  \caption{A randomized $1$-sensitive embedding of an arbitrary metric space $(X, d)$ into an $O(\log^2 k + \log k \log n)$-dimension $\ell_1$ space with $O(\log k)$ distortion} \label{fig:bourgain}
\end{figure}

\begin{theorem} [$1$-Sensitive Variant of Bourgain] \label{thm:bourgain}
  In the embedding given in Figure \ref{fig:bourgain}, for any data point $x \in \metric$ and any query point $y \in \query$, with probability at least $1 - \frac{1}{n^2 k}$, we have
  $$\frac{1}{64 \log k} d(x, y) \le \|\pi(x) - \pi(y)\|_1 \le d(x, y) \enspace.$$
\end{theorem}

The proof of the above theorem is very similar to one of the proofs for  Bourgain's original theorem. The expansion bound is identical. The contraction bound  will only guarantee the embedded distance of two pair of points $x \in \database$ and $y \in \query$ satisfies $d'(\pi(x), \pi(y)) \ge O(\frac{1}{\log k}) (d(x, y) - d(x, \query))$. We observe that the additive loss of $O(\frac{1}{\log k}) d(x, \query)$ can be avoided by using an additional $O(\log k + \log n)$ dimensions in the embedding. We include the proof below for completeness.

\begin{proof}
  \underline{\em Expansion:}~ By triangle inequality, $|\pi_{ij}(x) - \pi_{ij}(y)| = \frac{1}{K \log k} |d(x, S_{ij}) - d(y, S_{ij})| \le \frac{1}{K \log k} d(x, y)$. Summing over $0 \le i \le \log k$ and $1 \le j \le K$ we have $\sum_{i = 0}^{\log k} \sum_{j=1}^K |\pi_{ij}(x) - \pi_{ij}(y)| \le d(x, y)$.

  \medskip

  \noindent\underline{\em Contraction:}~ Let us first define some notation. Let $r_i$ and $r'_i$ denote the smallest radius such that the closed ball (with respect to~metric $(X, d)$, similar hereafter) $B(x, r_i)$ and $B(y, r'_i)$ respectively contains at least $2^{i-1}$ query points. Let $r^*_i = \max \{ r_i, r'_i \}$. We will have that $r^*_i$ is non-decreasing in $i$. Let $i'$ denote the largest index such that $r^*_{i'} + r^*_{i'-1} \le d(x, y)$. Redefine $r^*_{i'}$ to be $d(x, y) - r^*_{i'-1}$. We have $r^*_{i'} \ge \frac{d(x, y)}{2}$. We will need to following lemmas.

  \begin{lemma} \label{lem:2}
    For any $1 < i \le i'$, we have $\sum_{j=1}^K |\pi_{ij}(x) - \pi_{ij}(y)| \ge \frac{1}{32 \log k} \left( r^*_i - r^*_{i-1} \right)$ with probability at least $1 - \frac{1}{n^2 k \log k}$.
  \end{lemma}

  \begin{proof}
    Suppose $r^*_i = r_i$ (the other case is similar). Consider the open ball $B^o(x, r^*_i)$ and the closed ball $B(y, r^*_{i-1})$. By definition, the number of query points in $B^o(x, r^*_i)$ is less than $2^{i-1}$, and the number query points in $B(y, r^*_{i-1})$ is at least $2^{i-2}$. Since for each $1 \le j \le K$, the set $S_{ij}$ pick each query point independently with probability $2^{-(i-1)}$, the probability that $S_{ij} \cap B^o(x, r^*_i) = \emptyset$ is at least $(1-2^{-(i-1)})^{2^{i-1}} \ge \frac{1}{4}$, while the probability that $S_{ij} \cap B(y, r^*_{i-1}) \ne \emptyset$ is at least $1 - (1-2^{-(i-1)})^{2^{i-2}} \ge 1 - e^{-\frac{1}{2}}$. In sum, with probability at least $\frac{1}{4} (1 - e^{-\frac{1}{2}}) > \frac{1}{16}$, we have both $S_{ij} \cap B^o(x, r^*_i) = \emptyset$ and $S_{ij} \cap B(y, r^*_{i-1}) \ne \emptyset$, which indicates that $d(x, S_{ij}) \ge r^*_i$ and $d(y, S_{ij}) \le r^*_{i-1}$ and therefore
    \begin{equation}
      |\pi_{ij}(x) - \pi_{ij}(y)| \ge \frac{1}{K \log k} (r^*_i - r^*_{i-1}) \enspace. \label{eq:2}
    \end{equation}
    Further, by the additive form of Chernoff-Hoeffding theorem, we get that with probability at least $1 - 2^{-\frac{K}{64}} < 1 - \frac{1}{n^2 k \log k}$, \eqref{eq:2} holds for at least $\frac{K}{32}$ $i$'s. So we conclude that with probability at least $1 - \frac{1}{n^2 k \log k}$, $\sum_{j=1}^K |\pi_{ij}(x) - \pi_{ij}(y)| \ge \frac{1}{32 \log k} \left( r^*_i - r^*_{i-1} \right)$.
  \end{proof}

  \begin{lemma} \label{lem:3}
    $\sum_{j=1}^K |\pi_{1j}(x) - \pi_{1j}(y)| = \frac{1}{\log k} r^*_1$.
  \end{lemma}

  \begin{proof}
    It is easy to see that $r'_1 = 0$ because $y$ itself is a query point and $r^*_1 = r_1 = d(x, Q)$. Note that for every $j$, $S_{1j}$ equals the set of query points. So we always have $d(x, S_{1j}) = d(x, Q) = r^*_1$ and $d(y, S_{1j}) = 0$. Therefore, $|\pi_{1j}(x) - \pi_{1j}(y)| = \frac{1}{K \log k} r^*_1$ for $1 \le j \le K$, and summing up completes the proof.
  \end{proof}

  By Lemma \ref{lem:2} and union bound, with probability at least $1 - \frac{1}{n^2 k}$, we have
  $$\sum_{j=1}^K |\pi_{ij}(x) - \pi_{ij}(y)| \ge \frac{1}{32 \log k} \left( r^*_i - r^*_{i-1} \right)$$
  for all $1 < i \le i'$. By Lemma \ref{lem:3} we have
  $$\sum_{j=1}^K |\pi_{1j}(x) - \pi_{1j}(y)| \ge \frac{1}{\log k} r^*_1 \enspace.$$
  Summing them up we get that
  $$\sum_{i=1}^{\log k} \sum_{j=1}^K |\pi_{ij}(x) - \pi_{ij}(y)| \ge \frac{1}{32 \log k} r^*_{i'} \ge \frac{1}{64 \log k} d(x, y) \enspace.$$
\end{proof}

As a corollary of Theorem \ref{thm:bourgain} and union bound we have

\begin{corollary}
  In the embedding given in Figure \ref{fig:bourgain}, with probability at least $1 - \frac{1}{n}$, we have that for any data point $x \in D$ and any query points $y$,
  $$\frac{1}{64 \log k} d(x, y) \le \|\pi(x) - \pi(y)\|_1 \le d(x, y) \enspace.$$
\end{corollary}

Hence, there exists an embedding of an arbitrary metric to an $\ell_1$ metric with distortion $O(\log k)$ that is $1$-sensitive because it is oblivious to the data points. The dimension of the resulting $\ell_1$ metric is $O(\log^2 k + \log k \log n)$. We remark that the expansion guarantee may fail with some small probability, in which case the diameter of our embedding may be greater than $1$. This would appear to require us to move to an $(\epsilon,\delta)$-privacy guarantee, but it does not: when computing $\ell_1$ distances between points $x,y$, we can instead compute $\min(1,|\pi(x)-\pi(y)|_1)$. In the high probability event in which the expansion guarantee of the embedding holds, this will be exactly equal to the true distance between the embeddings of the points $x$ and $y$. In the small probability event in which the expansion guarantee fails, the resulting queries will remain $1/n$ sensitive in the private data. So by combining Theorem \ref{thm:accuracy}, Theorem \ref{thm:privacy}, and Theorem \ref{thm:bourgain} we have the following theorem.

\begin{theorem}
  For any metric space $(\metric, \dist)$, there is a non-interactive mechanism running in time poly$(n,k)$ for releasing answers to any $k$ distance queries with respect to~$(\metric, \dist)$ that is $(\eps, \delta)$-differentially private, such that with high probability it answers every query $y \in Q$ with accuracy
  \begin{align*}
    O\left(\frac{1}{\log k}\right) \frac{1}{n}\sum_{x \in D} d(x, y) - \bound{\tilde{O} \left( \frac{1}{n^{4/5} \eps^{4/5}}\right)} \le \mecapprox^{(\metric, \dist)} (y) \\
    \le \frac{1}{n}\sum_{x \in D} d(x, y) + \bound{\tilde{O} \left( \frac{1}{n^{4/5} \eps^{4/5}}\right)} \enspace.
  \end{align*}
  There is also a non-interactive mechanism running in time poly$(n,k)$ for releasing answers to any $k$ distance queries with respect to~$(\metric, \dist)$ that is $\eps$-differentially private, such that with high probability it answers every query $y \in Q$ with accuracy
  \begin{align*}
    O\left(\frac{1}{\log k}\right) \sum_{x \in D} d(x, y) - \bound{\tilde{O} \left( \frac{1}{n^{2/3} \eps^{2/3}} \right)} \le \mecpure^{(\metric, \dist)} (y) \\
    \le \sum_{x \in D} d(x, y) + \bound{\tilde{O} \left( \frac{1}{n^{2/3} \eps^{2/3}} \right)} \enspace.
  \end{align*}
\end{theorem}
\begin{remark}
Note that in this theorem, we require a dependence on $k$ both in the running time and in the accuracy bounds. This is because the embedding itself is a function of all of the queries in the query class. This is also what requires us to restrict attention to the non-interactive setting.
\end{remark}

\section{Conclusions}

We have shown that \emph{distance queries} defined over an arbitrary metric can be privately answered using efficient algorithms, circumventing known hardness results for less structured classes of linear queries. Our techniques crucially leveraged the metric structure of the queries, through our reliance on metric embeddings. Identifying other kinds of query structure that can be used to design efficient private query release algorithms remains one of the most important directions in differential privacy.

Towards this goal, we make a concrete conjecture. Let $\metric = [0,1]^\ell$ be the $\ell$-dimensional unit rectangle endowed with the Euclidean norm, and let $S \subseteq \{\phi:[0,1]^\ell\rightarrow [0,1]\}$ be the collection of predicates such that for each $\phi \in S$:
\begin{enumerate}
  \setlength{\partopsep}{0pt}
  \setlength{\topsep}{0pt}
  \setlength{\parsep}{0pt}
  \setlength{\itemsep}{0pt}
  \item $\phi$ is 1-Lipschitz: for all $x,y \in [0,1]^{\ell}$,  $|\phi(x) - \phi(y)| \leq ||x-y||_2$
  \item $\phi$ is convex: for all $x,y \in [0,1]^{\ell}$ and for all $t \in [0,1]$, $\phi(tx + (1-t)y) \leq t\phi(x) + (1-t)\phi(y)$
\end{enumerate}
For each $\phi \in S$, define the query $f_\phi(D) = \frac{1}{n}\sum_{x \in D}\phi(x)$. Then:
\begin{conjecture}
Let $C = \{f_\phi : \phi \in S\}$ denote the set of 1-Lipschitz, convex linear queries defined over the universe $\metric = [0,1]^\ell$. There is a differentially private query release mechanism operating in the interactive setting, that can answer any subset of $k$ queries from $C$ to additive error $\tilde{O}(\mathrm{poly}(\ell,\log(k))/\sqrt{n})$ with per-query update time $\textrm{poly}(\ell,n)$.
\end{conjecture}
Note that distance queries are a subset of convex, Lipschitz queries. Showing efficient algorithms for this entire set of queries would be an important step forwards towards the agenda of understanding the limitations of polynomial time private query release. We remark that if we remove the Lipschitz condition (and consider instead the class of all \emph{convex} queries), then this class includes \emph{boolean conjunctions}, which is already a challenge problem for efficient private query release. With the Lipschitz condition, this question is disjoint from (and possibly easier than) the question of efficiently releasing conjunctions.

\bibliographystyle{alpha}
\bibliography{privacy}



\end{document}